\documentclass[sigconf,authorversion]{acmart}
\usepackage[ruled, vlined]{algorithm2e}

\usepackage{float}
\usepackage{subfig}
\usepackage{caption}
\usepackage{csquotes}
\usepackage{xcolor}
\usepackage{fullpage}
\usepackage{lipsum}
\usepackage{pifont}
\usepackage{balance}

\usepackage{dblfloatfix}
\usepackage{booktabs} 

\newtheorem{problem}{Problem}
\newcommand\ilie{\textcolor{blue}}

\newcommand{\algname}{$\mathsf{ONBRA}$}

\usepackage{enumitem}
\newlist{todolist}{itemize}{2}
\setlist[todolist]{label=$\square$}
\usepackage{pifont}
\AtBeginDocument{%
  \providecommand\BibTeX{{%
    \normalfont B\kern-0.5em{\scshape i\kern-0.25em b}\kern-0.8em\TeX}}}

\copyrightyear{2022}
\acmYear{2022}
\setcopyright{acmlicensed}\acmConference[WWW '22]{Proceedings of the ACM Web Conference 2022}{April 25--29, 2022}{Virtual Event, Lyon, France}
\acmBooktitle{Proceedings of the ACM Web Conference 2022 (WWW '22), April 25--29, 2022, Virtual Event, Lyon, France}
\acmPrice{15.00}
\acmDOI{10.1145/3485447.3512204}
\acmISBN{978-1-4503-9096-5/22/04}

\begin{document}

\title{\algname: Rigorous Estimation of the Temporal Betweenness Centrality in Temporal Networks}

\author{Diego Santoro}
\authornote{Both authors contributed equally to this research.}
\affiliation{%
    \department{Department of Information Engineering}
    \institution{University of Padova}
    \city{Padova}
    \country{Italy}
}
\email{diego.santoro@phd.unipd.it}

\author{Ilie Sarpe}
\authornotemark[1]
\affiliation{%
    \department{Department of Information Engineering}
    \institution{University of Padova}
    \city{Padova}
    \country{Italy}
}
\email{sarpeilie@dei.unipd.it}

\begin{teaserfigure}
    \noindent\textit{\enquote{D\`{a}me `n' onbra de vin!}} -- A Venetian asking for some wine.
    \protect\
    \Description{This is not a figure, it is just a quote before the abstract. The quote is \enquote{D\`{a}me `n' onbra de vin!} which is a typical Venetian exclamation}
\end{teaserfigure}


\begin{abstract}
  In network analysis, the betweenness centrality of a node informally captures the fraction of shortest paths visiting that node. The computation of the betweenness centrality measure is a fundamental task in the analysis of modern networks, enabling the identification of the most central nodes in such networks. Additionally to being massive, modern networks also contain information about the \emph{time} at which their events occur. Such networks are often called \emph{temporal networks}. The temporal information makes the study of the betweenness centrality in temporal networks (i.e., \emph{temporal betweenness centrality}) much more challenging than in static networks (i.e., networks without temporal information). Moreover, the \emph{exact} computation of the temporal betweenness centrality is often impractical on even moderately-sized networks, given its extremely high computational cost. A natural approach to reduce such computational cost is to obtain high-quality \emph{estimates} of the exact values of the temporal betweenness centrality.
  In this work we present \algname, the first sampling-based approximation algorithm for estimating the temporal betweenness centrality values of the nodes in a temporal network, providing rigorous probabilistic guarantees on the quality of its output. \algname\ is able to compute the estimates of the temporal betweenness centrality values under two different optimality criteria for the shortest paths of the temporal network. In addition,  \algname\ outputs high-quality estimates with sharp theoretical guarantees leveraging on the \emph{empirical Bernstein bound}, an advanced concentration inequality.
  Finally, our experimental evaluation shows that \algname\ significantly reduces the computational resources required by the exact computation of the temporal betweenness centrality on several real world networks, while reporting high-quality estimates with rigorous guarantees.
\end{abstract}




\maketitle


\section{Introduction}
The study of centrality measures is a fundamental primitive in the analysis of networked datasets \cite{Borgatti2006, Newman2010}, and plays a key role in social network analysis~\cite{Das2018}. A centrality measure informally captures how important a node is for a given network according to \emph{structural} properties of the network. Central nodes are crucial in many applications such as analyses of co-authorship networks~\cite{Liu2005,Yan2009}, biological networks~\cite{Wuchty2003,Koschuetzki2008}, and ontology summarization~\cite{Zhang2007}. 

One of the most important centrality measures is the betweenness centrality~\cite{Freeman1977,Freeman1978}, which informally captures the fraction of \emph{shortest paths} going through a specific node. The betweenness centrality has found applications in many scenarios such as community detection~\cite{Fortunato2010}, link prediction~\cite{Ahmad2020}, and network vulnerability analysis~\cite{Holme2002}. The exact computation of the betweenness centrality of each node of a network is an extremely challenging task on modern networks, both in terms of running time and memory costs. Therefore, sampling algorithms have been proposed to provide provable high-quality approximations of the betweenness centrality values, while remarkably reducing the computational costs~\cite{Riondato2016,Riondato2018,Brandes2007}.

Modern networks, additionally to being large, have also richer information about their edges. In particular, one of the most important and 
easily accessible
information is the \emph{time} at which edges occur. Such networks are often called \emph{temporal networks}~\cite{Holme2019}. The analysis of temporal networks provides novel 
insights compared to the insights that 
would be obtained by the analysis of static networks (i.e., networks without temporal information), 
as, for example, in the study of subgraph patterns~\cite{Paranjape2017,Kovanen2011},  community detection \cite{Lehmann2019}, and network clustering~\cite{Fu2020}. As well as for static networks, the study of the temporal betweenness centrality in temporal networks aims at identifying the nodes  
that are visited by a high number of \emph{optimal} paths ~\cite{Holme2012,Buss2020}. 
In temporal networks, the definition of optimal paths has to consider the information about the timing of the edges, making the possible definitions of optimal paths much more richer than in static networks~\cite{Rymar2021}. 

In this work, a temporal path is valid if it is time respecting, i.e. if all the interactions within the path occur at increasing timestamps (see Figures \ref{subfig:staticPath}-\ref{subfig:temporalShortestPath}). We considered two different optimality criteria for temporal paths, chosen for their relevance~\cite{Holme2012}: (i) shortest temporal path (STP) criterion, a commonly used criterion for which a path is optimal if it uses the minimum number of interactions to connect a given pair of nodes; (ii) restless temporal path (RTP) criterion, for which a path is optimal if, in addition to being shortest, all its consecutive interactions occur at most within a given user-specified time duration parameter $\delta\in\mathbb{R}$ (see Figure \ref{subfig:temporalShortestPath}).
The RTP criterion finds application, for example, in the study of spreading processes over complex networks~\cite{Pan2011}, where information about the timing of consecutive interactions is fundamental. 
The exact computation of the temporal betweenness centrality under 
the STP and RTP optimality criteria becomes impractical (both in terms of running time and memory usage) for even moderately-sized networks. Furthermore, as well as for static networks, obtaining a high-quality approximation of the temporal betweenness centrality of a node is often sufficient in many applications. Thus, we propose \algname, the \emph{first} algorithm to compute rig\underline{O}rous estimatio\underline{N} of temporal \underline{B}etweenness cent\underline{R}ality values in tempor\underline{A}l networks\footnote{\url{https://vec.wikipedia.org/wiki/Onbra}.}, providing sharp guarantees on the quality of its output. As for many data-mining algorithms, \algname's output is function of two parameters: $\varepsilon \in (0,1)$ controlling the estimates' accuracy; and $\eta \in (0,1)$ controlling the confidence. 
The algorithmic problems arising from accounting for temporal information are really challenging to deal with compared to the static network scenario, although \algname\ shares a high-level sampling strategy similar to~\cite{Riondato2018}.
Finally, we show that in practice our algorithm \algname, other than providing high-quality estimates while reducing computational costs, it also enables analyses that cannot be otherwise performed with existing state-of-the-art algorithms. Our main contributions are the following:
\begin{itemize}
    \item We propose \algname, the first sampling-based algorithm that outputs high-quality approximations of the temporal betweenness centrality values of the nodes of a temporal network. \algname\ leverages on an advanced data-dependent and variance-aware concentration inequality to provide sharp probabilistic guarantees on the quality of its estimates. 
    \item We show that \algname\ is able to compute high-quality temporal betweenness estimates for two optimality criteria of the paths, i.e., STP and RTP criteria. In particular, we developed specific 
    algorithms for \algname\ to address the 
    computation of the estimates 
    according to 
    such optimality criteria.
    \item We perform an extensive experimental evaluation with several goals: (i) under the STP criterion,  show that studying the temporal betweenness centrality provides novel insights compared to the static version; (ii) under the STP criterion, show that \algname\ provides high-quality estimates, while significantly reducing the computational costs
     compared to the state-of-the-art exact algorithm, and that it enables the study of large datasets that cannot practically be analyzed by the existing exact algorithm; (iii) 
      show that \algname\ is able to estimate the temporal betweenness centrality under the RTP optimality criterion by varying $\delta$.
\end{itemize}

\begin{figure}[t]
	\centering
	\subfloat[]{
		\begin{tabular}{lr}
			\includegraphics[width=.36\linewidth]{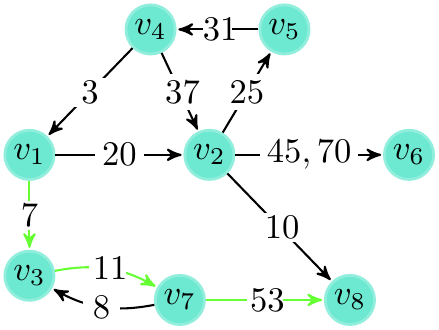} &  \includegraphics[width=.36\linewidth]{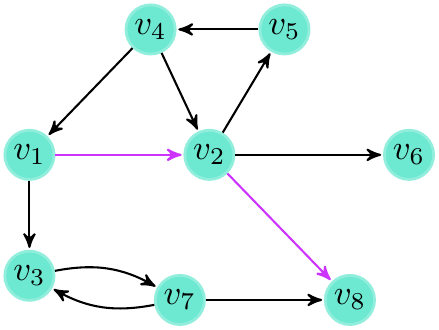} 
		\end{tabular}
		\label{subfig:temporalNetwork}
	}\\
	\subfloat[]{
		\includegraphics[width=.35\linewidth]{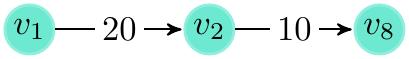}
		\label{subfig:staticPath}
	}
	\subfloat[]{
		\includegraphics[width=.45\linewidth]{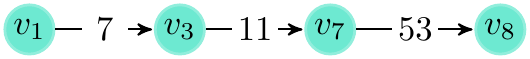}
		\label{subfig:temporalShortestPath}
	}
	\caption{(\ref{subfig:temporalNetwork}): (left) a temporal network $T$ with $n=8$ nodes and $m=12$ edges, (right) its associated static network $G_T$ obtained from $T$ by removing temporal information. A shortest \emph{temporal} path cannot be identified by a shortest path in the static network: e.g., the shortest paths from node $v_1$ to node $v_8$, respectively coloured in green in $T$ and purple in $G_T$, are different. (\ref{subfig:staticPath}): A path that is not time respecting. (\ref{subfig:temporalShortestPath}): A time respecting path that is also shortest in $T$. With $\delta\ge 42$ such path is also shortest $\delta$-restless path.}
	\label{fig:basicdef}
\end{figure}

\section{Preliminaries}
In this section we introduce the fundamental notions needed throughout the development of our work and formalize the problem of approximating the temporal betweenness centrality of the nodes in a temporal network. 

We start by introducing temporal networks. 
\begin{definition}
A \emph{temporal network} $T$ is a pair $T=(V,E)$, where $V$ is a set of $n$ nodes (or vertices), and $E=\{(u,v,t): u,v\in V, u\neq v, t \in \mathbb{R}^+\}$ is a set of $m$ directed edges\footnote{\algname\ can be easily adapted to work on \emph{undirected} temporal networks with minor modifications.}\footnote{W.l.o.g.\ we assume the edges $(u_1,v_1,t_1),\dots,(u_m,v_m,t_m)$ to be sorted by increasing timestamps.}.
\end{definition}

Each edge $e = (u,v,t) \in E$ of the network represents an interaction from node $u \in V$ to node $v \in V$ at time $t$, which is the \emph{timestamp} of the edge. Figure \ref{subfig:temporalNetwork} (left) provides an example of a temporal network $T$.
Next, we define \emph{temporal paths}.
\begin{definition}
Given a temporal network $T$, a \emph{temporal path} $\mathsf{P}$ is a sequence $\mathsf{P}=\langle  e_1=(u_1,v_1,t_1), e_2=(u_2,v_2,t_2), \dots,e_k=(u_k,v_k,t_k) \rangle$ of $k$ edges of $T$ ordered by increasing timestamps\footnote{Our work can be easily adapted to deal with non-strict ascending timestamps (i.e., with $\leq$ constraints).}, i.e., $t_i < t_{i+1}, i \in \{1,\dots, k-1\}$, such that the node $v_i$ of edge $e_i$ is equal to the node $u_{i+1}$ of the consecutive edge $e_{i+1}$, i.e., $v_i=u_{i+1},  i \in \{1,\dots,k-1\}$, and each node $v \in V$ is visited by $\mathsf{P}$ \emph{at most} once. 
\end{definition}
Given a temporal path $\mathsf{P}$ made of $k$ edges, we define its length as $\ell_{\mathsf{P}} = k$. An example of temporal path $\mathsf{P}$ of length $\ell_{\mathsf{P}} = 3$ is given by Figure \ref{subfig:temporalShortestPath}. 
Given a \emph{source} node $s \in V$ and a \emph{destination} node $z \in V, z\neq s$, a \emph{shortest} temporal path between $s$ and $z$ is a temporal path $\mathsf{P_{s,z}}$ of length $\ell_{\mathsf{P}_{s,z}}$ such that in $T$ there is no temporal path $\mathsf{P}_{s,z}'$ connecting $s$ to $z$ of length $\ell_{\mathsf{P}_{s,z}'}<\ell_{\mathsf{P}_{s,z}}$. Given a temporal shortest path $\mathsf{P}_{s,z}$ connecting $s$ and $z$, we define $\mathtt{Int}(\mathsf{P}_{s,z}) = \{w \in V | \ \exists \ (u,w,t) \lor (w,v,t) \in \mathsf{P}_{s,z}, w \neq s,z\} \subset V$ as the set of nodes \emph{internal} to the path $\mathsf{P}_{s,z}$. Let $\sigma_{s,z}^{sh} $ be the number of shortest temporal paths between  nodes $s$ and $z$. Given a node $v \in V$, we denote with $\sigma_{s,z}^{sh} (v)$ the number of shortest temporal paths $\mathsf{P}_{s,z}$ connecting $s$ and $z$ for which $v$ is an internal node, i.e., $\sigma_{s,z}^{sh} (v) = |\{\mathsf{P}_{s,z} | v \in \mathtt{Int}(\mathsf{P}_{s,z})\}|$. Now we introduce the \emph{temporal betweenness centrality}  of a node $v \in V$, which intuitively captures the fraction of shortest temporal paths visiting $v$.
\begin{definition}
We define the \emph{temporal betweenness centrality} $b(v)$ of a node $v \in V$ as
\[
    b(v) = \frac{1}{n(n-1)} \sum_{s,z\in V, \ s\neq z} \frac{\sigma_{s,z}^{sh} (v)}{\sigma_{s,z}^{sh} }.
\]
\end{definition}

Let $B(T) = \{(v,b(v)): v \in V\}$ be the set of pairs composed of a node $v \in V$ and its temporal betweenness value $b(v)$. Since the exact computation of the set $B(T)$ using state-of-the-art exact algorithms, e.g., \cite{Buss2020,Rymar2021}, is impractical on even moderately-sized temporal networks (see Section \ref{sec:exp} for experimental evaluations), 
in our work we aim at providing high-quality approximations of the temporal betweenness centrality values of all the nodes of the temporal network. That is, we compute the set $\tilde{B}(T)= \{(v,\tilde{b}(v)): v \in V\}$, where $\tilde{b}(v)$ is an accurate estimate of $b(v)$, controlled by two parameters $\varepsilon, \eta \in (0,1)$, (accuracy and confidence). We want $\tilde{B}(T)$ to be an \emph{absolute ($\epsilon,\eta$)-approximation set} of $B(T)$, as commonly adopted in data-mining algorithms (e.g., in  \cite{Riondato2018}): that is, $\tilde{B}(T)$ is an approximation set such that
\[
    \mathbb{P}\left[\sup_{v \in V}|\tilde{b}(v)- b(v)| \le \varepsilon\right] \geq 1 - \eta.
\]
Note that in an absolute ($\epsilon,\eta$)-approximation set, for each node $v \in V$, the estimate $\tilde{b}(v)$ of the temporal betweenness value deviates from the actual value $b(v)$ of at most $\varepsilon$, with probability at least $1-\eta$. Finally, let us state the main computational problem addressed in this work.

\begin{problem}
	\label{problem}
    Given a temporal network $T$ and two parameters $(\varepsilon, \eta)\in (0,1)^2$, compute the set  $\tilde{B}(T)$, i.e., an absolute ($\epsilon,\eta$)-approximation set of $B(T)$.
\end{problem}


\section{Related Works}\label{sec:relwork}
Given the importance of the betweenness centrality for network analysis, many algorithms have been proposed to compute it in different scenarios. In this section we focus on those scenarios  most relevant to our work, grouped as follows.

\emph{Approximation Algorithms for Static Networks.} 
Recently, many algorithms to approximate the betweenness centrality in static networks have been proposed, most of them employ randomized sampling approaches~\cite{Riondato2016,Riondato2018,Brandes2007}. The existing algorithms differ from each other mainly for the sampling strategy they adopt and for the probabilistic guarantees they offer. Among these works, the one that shares similar ideas to our work is \cite{Riondato2018} by Riondato and Upfal, where the authors proposed to sample pairs of nodes $(s,z)\in V^2$, compute all the shortest paths from $s$ to $z$, and update the estimates of the betweenness centrality values of the nodes internal to such paths. The authors developed a suite of algorithms to output an $(\varepsilon,\eta)$-approximation set of the set of betweenness centrality values.
Their work cannot be easily adapted to temporal networks. In fact, static and temporal paths in general are not related in any way, and the temporal scenario introduces many novel  challenges: (i) computing the optimal temporal paths, and (ii) updating the betweenness centrality values. 
Therefore, our algorithm \algname\ employs the idea of the estimator provided by \cite{Riondato2018}, while using novel algorithms designed for the context of temporal networks. Furthermore, the probabilistic guarantees provided by our algorithm \algname\ leverage on the variance of the estimates, differently from~\cite{Riondato2018} that used bounds based on the Rademacher averages. 
Our choice to use a variance-aware concentration inequality is motivated by the recent interest in
 providing sharp guarantees employing the \emph{empirical variance} of the estimates~\cite{Cousins2021, Pellegrina2021}.

\emph{Algorithms for Dynamic Networks.} In this setting the algorithm keeps track of the betweenness centrality value of each node for every timestamp $t_1,\dots,t_m$ observed in the network~\cite{Lee2012,Hanauer2021}.
Note that this is extremely different from estimating the temporal betweenness centrality values in temporal networks.
In the dynamic scenario the paths considered are \emph{not} required to be time respecting. For example, in the dynamic scenario, if we consider the network in Figure \ref{subfig:temporalNetwork} (left) at any time $t>20$, the shortest path from $v_1$ to $v_8$ is the one highlighted in purple in Figure \ref{subfig:temporalNetwork} (right). Instead, in the temporal setting such path is not time respecting. 
We think that it is very challenging to adapt the algorithms for dynamic networks to work in the context of temporal networks, which further motivates us to propose
\algname. 

\emph{Exact Algorithms for Temporal Networks.} Several exact approaches have been proposed in the literature~\cite{Tsalouchidou2020,Alsayed2015,Kim2012}. 
The algorithm most relevant to our work was presented in~\cite{Buss2020}, where
the authors extended the well-known Brandes algorithm~\cite{Brandes2001} to the temporal network scenario considering the STP criterion (among several other criteria). 
They showed that the time complexity of their algorithm is $O(n^3(t_{m}-t_{1})^2)$, which is often impractical on even moderately-sized networks. Recently, \cite{Rymar2021} discussed conditions on temporal paths under which the temporal betweenness centrality can be computed in polynomial time, showing a general algorithm running in $O(n^2m(t_m-t_1)^2)$ even under the RTP criterion, which is again very far from being  practical on modern networks. 

We conclude by observing that, to the best of our knowledge, no approximation algorithms exist for estimating the temporal betweenness centrality in temporal networks.
\section{Method, Algorithm, and Analysis}
\label{sec:method_algorithm_analysis}
In this section we discuss \algname, our novel algorithm for computing high-quality approximations of the temporal betweenness centrality values of the nodes of a temporal network. We first discuss the sampling strategy used in \algname, then we present the algorithm, and finally we show the theoretical guarantees on the quality of the estimates of \algname. 

\subsection{\algname\ - Sampling Strategy} In this section we discuss the sampling strategy adopted by \algname\ that is independent of the optimality criterion of the paths. However, for the sake of presentation, we discuss the sampling strategy for the STP-based temporal betweenness centrality estimation. 

\algname\ samples \emph{pairs} of nodes $(s,z)$ and computes all the shortest temporal paths from $s$ to $z$. More formally, let $\mathcal{D} = \{(u,v)\in V^2: u\neq v \}$, and $\ell \in \mathbb{N}, \ell\ge2$ be a user-specified parameter. \algname\ first collects $\ell$ pairs of nodes $(s_i,z_i)_i,i=1,\dots,\ell$, sampled uniformly at random from $\mathcal{D}$. Next, for each pair $(s,z)$ it computes $\mathcal{P}_{s,z}=\{\mathsf{P}_{s,z}: \mathsf{P}_{s,z} \text{ is shortest}\}$, i.e., the set of shortest temporal paths from $s$ to $z$. Then, for each node $v \in V$ s.t.\ $\exists \mathsf{P}_{s,z} \in \mathcal{P}_{s,z}$ with $v \in \mathtt{int}(\mathsf{P}_{s,z})$, i.e., for each node $v$ that is internal to a shortest temporal path of $\mathcal{P}_{s,z}$, \algname\ computes the estimate $\tilde{b}'(v)= \sigma_{s,z}^{sh}(v) / \sigma_{s,z}^{sh}$, which is an unbiased estimator of the temporal betweenness centrality value $b(v)$ (i.e., $\mathbb{E}[\tilde{b}'(w)] = b(v)$, see Lemma \ref{lemma:unbiased} in Appendix \ref{app:proofs}). 
Finally, after processing the $\ell$ pairs of nodes randomly selected, \algname\ computes for each node $v \in V$ the (unbiased) estimate $\tilde{b}(v)$ of the actual temporal betweenness centrality $b(v)$ by averaging $\tilde{b}'(v)$ over the $\ell$ sampling steps: $\tilde{b}(v) = 1/\ell \sum_{i=1}^\ell \tilde{b}'(v)_i $, where $\tilde{b}'(v)_i$ is the estimate of $b(v)$ obtained by analyzing the $i$-th sample, $i\in [1,\ell]$. We will discuss the theoretical guarantees on the quality of $\tilde{b}(v)'s$ in Section \ref{subsec:theoguar}.


\begin{algorithm}[t]
    \DontPrintSemicolon
    \LinesNumbered 
    \SetKwInOut{Input}{input}
    \SetKwInOut{Output}{output}
    \SetKwFunction{GetSample}{uniformRandomSample}
    \SetKwFunction{ModifiedSP}{SourceDestinationSTPComputation}
    \SetKwFunction{UpdateEST}{updateSTPEstimates}
    \SetKwComment{tcp}{//}{}
    \KwIn{Temporal network $T=(V,E)$, $\eta\in(0,1), \ell \ge 2$}
    \KwOut{Pair $(\varepsilon',\tilde{B}(T))$ s.t. $\tilde{B}$ is an absolute $(\varepsilon', \eta)$-approximation set of $B(T)$.}
    $\mathcal{D}\leftarrow\{(u,v)\in V\times V, u\neq v\}$\label{algline:samplespace}\;
    $\tilde{B}_{v,:} \gets \vec{0}_\ell, \forall v\in V$\label{algline:estimatesmatrix}\;
    \For{$i\gets 1$ \KwTo $\ell$} {\label{algline:mainloop}
        $(s,z)\leftarrow$ \GetSample{$\mathcal{D}$}\label{algline:getsample}\;
        \ModifiedSP{$T,s,z$}\label{algline:spcomp}\;
        \If{\emph{reached}$(z)$\label{algline:ifReachedDest}} {
            \UpdateEST{$\tilde{B}, i$}\label{algline:updateEst}\;
        }
    }
    $\tilde{B}(T) \gets \{(v, 1/\ell \sum_{i=1}^\ell \tilde{B}_{v,i}): v\in V \} $\label{algline:unbiasedComp}\;
    $\varepsilon' \gets \sup_{v\in V} \left\{ \sqrt{\frac{ 2 \mathbf{V}(\tilde{B}_{v,:}) \ln(4n / \eta)}{\ell}} + \frac{7 \ln(4n / \eta)}{3(\ell - 1)} \right\} $\label{algline:compStoppingCond}\;
    
    \Return{$(\varepsilon', \tilde{B}(T))$}\;
    \caption{\algname.}
    \label{alg:static}
\end{algorithm}

\subsection{Algorithm Description}
\subsubsection*{Sampling Algorithm: \algname} \algname\ is presented in Algorithm \ref{alg:static}. In line \ref{algline:samplespace} we first initialize the set $\mathcal{D}$ of objects to be sampled, where each object is a pair of distinct nodes from $V$. Next, in line \ref{algline:estimatesmatrix} we initialize the matrix $\tilde{B}$ of size $|V| \cdot \ell$ to store the estimates of \algname\ for each node at the various iterations, needed to compute their empirical variance and the final estimates. Then we start the main loop (line \ref{algline:mainloop}) that will iterate $\ell$ times. In such loop we first select a pair $(s,z)$ sampled uniformly at randomly from $\mathcal{D}$ (line \ref{algline:getsample}). We then compute all the shortest temporal paths from $s$ to $z$ by executing Algorithm \ref{alg:truncatedstpaths} (line \ref{algline:spcomp}), which is described in detail later in this section. Such algorithm computes all the shortest temporal paths from $s$ and $z$ adopting some pruning criteria to speed-up the computation. If at least one STP between $s$ and $z$ exists (line \ref{algline:ifReachedDest}), then for each node $v \in V$ internal to a path in $\mathcal{P}_{s,z}$ we update the corresponding estimate to the current iteration by computing $\tilde{b}'(v)_i$ using Algorithm \ref{alg:updatepathcounts} (line \ref{algline:updateEst}). While in static networks this step can be done with a simple recursive formula \cite{Riondato2018}, in our scenario we need a specific algorithm to deal with the more challenging fact that a node may appear at different distances from a given source across different shortest temporal paths. We will discuss in detail such algorithm later in this section. 
At the end of the $\ell$ iterations of the main loop, \algname\ computes: (i) the set $\tilde{B}(T)$ of unbiased estimates (line \ref{algline:unbiasedComp}); (ii) and a tight bound $\varepsilon'$ on $\sup_{v \in V}|\tilde{b}(v)- b(v)|$, which leverages the empirical variance $\mathbf{V}(\tilde{B}_{v,:})$ of the estimates (line \ref{algline:compStoppingCond}). We observe that $\varepsilon'$ is such that the set $\tilde{B}(T)$ is an absolute $(\varepsilon', \eta)$-approximation set of $B(T)$. We discuss the computation of such bound in Section \ref{subsec:theoguar}. Finally, \algname\ returns $(\varepsilon',\tilde{B}(T))$.

\begin{algorithm}[t]
    \DontPrintSemicolon
    \SetKwComment{Comment}{$\triangleright$\ }{}
    \LinesNumbered 
    \KwIn{$T=(V,E)$, source node $s$, destination node $z$}
    \For{$v\in V$}
    {
        $\mathsf{dist}_{v} \gets -1$; $\sigma_{v} \gets 0$\label{alglinetr:structdist}\;
    }
    \For{$(u,v,t)\in E$}
    {
        $\sigma_{v,t}\gets 0; P_{v,t}\gets \emptyset; \mathsf{dist}_{v,t} \gets -1$\label{alglinetr:structpaths}\;
    }
    $\mathsf{dist}_s \gets 0$; $\mathsf{dist}_{s,0} \gets 0$\label{alglinetr:initdist}\;
    $\sigma_s \gets 1$; $\sigma_{s,0} \gets 1$; $\mathsf{d}_z^{\min} \gets \infty$\label{alglinetr:initpaths}\;
    $Q \gets$ empty queue; $Q.$enqueue$((s,0))$\label{alglinetr:initdatastruct}\;
    \While{$!Q.$\emph{empty}$()$\label{alglinetr:mainwhile}}
    {
        $(v,t) \gets Q.$dequeue$()$ \label{alglinetr:curpair}\;
        \If{$(\mathsf{dist}_{v,t} < \mathsf{d}_z^{\min})$\label{alglinetr:truncCond}}{
            \For{$(w,t')\in \mathcal{N}^{{+}}(v,t), $\label{alglinetr:forOutNeighTemp}}
            {
                \If{$\mathsf{dist}_{w,t'}=-1$\label{alglinetr:ifNeverReachedAtTime}}
                {
                    $\mathsf{dist}_{w,t'} \gets \mathsf{dist}_{v,t}+1$\label{alglinetr:mindistTime}\; 
                    {\If{$\mathsf{dist}_w=-1$\label{alglinetr:ifNeverReached}}
                        {
                            $\mathsf{dist}_w \gets \text{dist}_{v,t}+1$\label{alglinetr:mindDistToNode}\; 
                            \If{$w=z$}
                            {
                                $\mathsf{d}_z^{\min} \gets \mathsf{dist}_w$\label{alglinetr:mindDistToDest}\;
                            }
                        }     
                    }
                    $Q.$enqueue$((w,t'))$\label{alglinetr:enqueue}\;
                }
                \If{$\mathsf{dist}_{w,t'}= \mathsf{dist}_{v,t}+1$\label{alglinetr:ifShortest}}
                {
                    $\sigma_{w,t'} \gets \sigma_{w,t'} + \sigma_{v,t}$\label{alglinetr:updatePaths}\;
                    $P_{w,t'} \gets  P_{w,t'} \cup \{(v,t)\}$\label{alglinetr:updatePredecessors}\;
                    {\If{$\mathsf{dist}_{w,t'}=\mathsf{dist}_w$}
                        {
                            $\sigma_w \gets \sigma_w+\sigma_{v,t}$\label{alglinetr:updateTotPaths}\; 
                    }}
                }
                
            }
        }
        
    }
    
    \caption{Source-Destination STP computation.}
    \label{alg:truncatedstpaths}
\end{algorithm} 

\subsubsection*{Subroutines} We now describe the subroutines employed in Algorithm \ref{alg:static} focusing on the STP criterion. Then, in Section \ref{sec:RTP_criteria}, we discuss how to deal with the RTP criterion. 

\emph{Source Destination Shortest Paths Computation.} 
We start by introducing some definitions needed through this section. 
First, we say that a pair $(v,t)\in V\times \{t_1,\dots,t_m\}$ is a \emph{vertex appearance} (VA) if $\exists (u,v,t)\in E$. 
Next, given a VA\ $(v,t)$ we say that a VA $(w,t')$ is a \emph{predecessor} of $(v,t)$ if $\exists(w,v,t)\in E, t' < t$. 
Finally, given a VA $(v,t)$ we define its set of \emph{out-neighbouring VAs} as $\mathcal{N}^{{+}}(v,t)=\{(w,t') : \exists (v,w,t')\in E, t<t'\}$.

We now describe Algorithm \ref{alg:truncatedstpaths} that computes the shortest temporal paths between a source node $s$ and a destination node $z$ (invoked in \algname\ at line \ref{algline:spcomp}). 
Such computation  is optimized to prune the search space once found the destination $z$.
The algorithm initializes the data structures needed to keep track of the shortest temporal paths that, starting from $s$, reach a node in $V$, i.e., the arrays $\mathsf{dist}[\cdot]$ and $\sigma[\cdot]$ that contain for each node $v\in V$, respectively, the minimum distance to reach $v$ and the number of shortest temporal paths reaching $v$ (line \ref{alglinetr:structdist}). 
In line \ref{alglinetr:structpaths} we initialize $\mathsf{dist}[\cdot,\cdot]$ that keeps track of the minimum distance of a VA from the source $s$, $\sigma[\cdot, \cdot]$ that maintains the number of shortest temporal paths reaching a VA from $s$, and $P$ keeping the set of predecessors of a VA across the shortest temporal paths explored. 
After initializing the values of the data structures for the source $s$ and $\mathsf{d}_z^{\min}$ keeping the length of the minimum distance to reach $z$ (lines \ref{alglinetr:initdist}-\ref{alglinetr:initpaths}), we initialize the queue $Q$ that keeps the VAs to be visited in a BFS fashion in line \ref{alglinetr:initdatastruct} (observe that, since the temporal paths need to be time-respecting, all the paths need to account for the time at which each node is visited). 
Next, the algorithm explores the network in a BFS order (line \ref{alglinetr:mainwhile}), extracting a VA $(v,t)$ from the queue, which corresponds to a node and the time at which such node is visited, and processing it by collecting its set $\mathcal{N}^{{+}}(v,t)$ of out-neighbouring VAs  (lines \ref{alglinetr:curpair}-\ref{alglinetr:forOutNeighTemp}).
If a VA $(w,t')$ was not already explored (i.e., it holds $\mathsf{dist}_{w,t'}=-1$), then we update the minimum distance $\mathsf{dist}_{w,t'}$ to reach $w$ at time $t'$, the minimum distance $\mathsf{dist}_{w}$ of the vertex $w$ if it was not already visited, and, if $w$ is the destination node $z$, we update $\mathsf{d}_z^{\min}$ (lines \ref{alglinetr:ifNeverReachedAtTime}-\ref{alglinetr:mindDistToDest}). 
Observe that the distance $\mathsf{d}_z^{\min}$ to reach $z$ is used as a \emph{pruning criterion} in line \ref{alglinetr:truncCond} (clearly, if a VA appears at a distance greater than $\mathsf{d}_z^{\min}$ then it cannot be on a shortest temporal path from $s$ to $z$). 
After updating the VAs to be visited by inserting them in $Q$ (line \ref{alglinetr:enqueue}), if the current temporal path is shortest for the VA $(w,t')$  analyzed, we update the number $\sigma_{w,t'}$ of shortest temporal paths leading to it, its set $P_{w,t'}$ of predecessors, and the number $\sigma_w$ of shortest temporal paths reaching the node $w$ (lines \ref{alglinetr:ifShortest}-\ref{alglinetr:updateTotPaths}). 

\emph{Update Estimates: STP criterion.} Now we describe Algorithm \ref{alg:updatepathcounts}, which updates the temporal betweenness estimates of each node internal to a path in $\mathcal{P}_{s,z}$ already computed. 
With Algorithm \ref{alg:truncatedstpaths} we computed for each VA $(w,t)$ the number $\sigma_{w,t}$ of shortest temporal paths from $s$ reaching $(w,t)$.
Now, in Algorithm \ref{alg:updatepathcounts} we need to combine such counts to compute the total number of shortest temporal paths leading to each VA $(w,t)$ appearing in a path in $\mathcal{P}_{s,z}$, allowing us to compute the estimate of \algname\ for each node $w$.

At the end of Algorithm \ref{alg:truncatedstpaths} there are in total $|\mathcal{P}_{s,z}|$ shortest temporal paths reaching $z$ from $s$. 
Now we need to compute, for each node $w$ internal to a path in $\mathcal{P}_{s,z}$ and for each VA $(w,t)$, the number $\sigma_{w,t}^z$ of shortest temporal paths leading from $w$ to $z$ at a time greater that $t$.
Then, the fraction of paths containing the node $w$ is computed with a simple formula, i.e., $\sum_t \sigma_{w,t}^z\cdot \sigma_{w,t}/\sigma_z$, where $\sigma_z =|\mathcal{P}_{s,z}| $. The whole procedure is described in Algorithm \ref{alg:updatepathcounts}. 
We start by initializing $\sigma_{v,t}^z$ that stores for each VA $(v,t)$ the number of shortest temporal paths reaching $z$ at a time greater than $t$ starting from $v$, and a boolean matrix $M$ that keeps track for each VA if it has already been considered (line \ref{alglineup:initdatastruct}). 
In line \ref{alglineup:initqueue} we initialize a queue $R$ that will be used to explore the VAs appearing along the paths in $\mathcal{P}_{s,z}$ in reverse order of distance from $s$ starting from the destination node $z$. 
Then we initialize $\sigma^z_{w,t'}$ for each VA reaching $z$ at a given time $t'$ (line \ref{alglineup:initsigmas}), and we insert each VA in the queue only one time (line \ref{alglineup:initVappeareance}). 
The algorithm then starts its main loop exploring the VAs in decreasing order of distance starting from $z$ (line \ref{alglineup:mainwhile}). 
We take the VA $(w,t)$ to be explored in line \ref{alglineup:dequeue}. 
If $w$ differs from $s$ (i.e., $w$ is an internal node), then we update its temporal betweenness estimate by adding $\sigma_{w,t}^z\cdot \sigma_{w,t}/\sigma_z$ (line \ref{alglineup:updateBetweenness}). 
As we did in the initialization step, then we process each predecessor $(w',t')$ of $(w,t)$ across the paths in $\mathcal{P}_{s,z}$ (line \ref{alglineup:forPred}), update the count $\sigma_{w',t'}^z$ of the paths from the predecessor to $z$ by summing the number $\sigma_{w,t}^z$ of paths passing through $(w,t)$ and reaching $z$ (line \ref{alglineup:UpdatePaths}), and we enqueue the predecessor $(w',t')$ only if it was not already considered (lines \ref{alglineup:inqueuewhile}-\ref{alglineup:enqueue}). 
So, the algorithm terminates by having properly computed for each node $v\in V, v\neq s,z$ the estimate $\tilde{b}'(v)_i$ for each iteration $i\in[1,\ell]$.

\begin{algorithm}[t]
    \DontPrintSemicolon
    \SetKwComment{Comment}{$\triangleright$\ }{}
    \LinesNumbered 
    \KwIn{$\tilde{B}, i$.}
    \For{$(u,v,t)\in E$}
    {
        $\sigma_{v,t}^z\gets 0$; $M_{v,t} \gets \mathsf{False}$\label{alglineup:initdatastruct}\;
    }
    $R \gets$ empty queue; \label{alglineup:initqueue}\;
    \ForEach{$t : (\sigma_{z,t} > 0)$\label{alglineup:fortReachingZ}}
    {
        \For{$(w,t')\in P_{z,t}$\label{alglineup:forPredecessorsZ}}
        {
            $\sigma_{w,t'}^z \gets \sigma_{w,t'}^z+1$\label{alglineup:initsigmas}\;
            \If{$!M_{w,t'}$\label{alglineup:inqueue}}{
                $R$.enqueue($(w,t')$); $M_{w,t'} \gets \mathsf{True} $\label{alglineup:initVappeareance}\;           
            }
        }
    }
    \While{$!R.$\emph{empty}$()$\label{alglineup:mainwhile}}{
        $(w,t) \gets R.$dequeue$()$\label{alglineup:dequeue}\;
        \If{$w\neq s$\label{alglineup:nots}}{
            $\tilde{B}_{w,i} \gets \tilde{B}_{w,i} + \sigma_{w,t}^z \cdot \sigma_{w,t} / \sigma_z$\label{alglineup:updateBetweenness}\;
            \For{$(w',t')\in P_{w,t}$\label{alglineup:forPred}}
            {
                $\sigma_{w',t'}^z \gets \sigma_{w',t'}^z + \sigma_{w,t}^z$\label{alglineup:UpdatePaths}\;
                \If{$!M_{w',t'}$\label{alglineup:inqueuewhile}}{ 
                    $R$.enqueue($(w',t')$); $M_{w',t'} \gets \mathsf{True} $\label{alglineup:enqueue}\;           
                }
            }
        }
        
    }

    \caption{Update betweenness estimates - STP.}
    \label{alg:updatepathcounts}
\end{algorithm} 

\subsection{Restless Temporal Betweenness}
\label{sec:RTP_criteria}
In this section we present the algorithms that are used in \algname\ 
when considering the RTP criterion for the optimal paths to compute the temporal betweenness centrality values.

Recall that, in such scenario, a temporal path $\mathsf{P}=\langle  e_1=(u_1,v_1,t_1), \ e_2=(u_2,v_2,t_2), \dots,e_k=(u_k,v_k,t_k) \rangle$ is considered optimal if and only if $\mathsf{P}$, additionally to being \emph{shortest}, is such that, given $\delta \in \mathbb{R}^+$, it holds  $t_{i+1}\le t_{i}+\delta$ for $i=1,\dots,k-1$.
Considering the RTP criteria, we need to relax the definition of shortest temporal paths and, instead, consider \emph{shortest temporal walks}. Intuitively, a walk is a path where we drop the constraint that a node must be visited at most once. We provide an intuition of why we need such requirement in Figure \ref{fig:twalks}. Given $\delta \in \mathbb{R}^+$, we refer to a shortest temporal walk as \emph{shortest $\delta$-restless temporal walk}.

In order to properly work under the RTP criteria, \algname\ needs novel algorithms to compute the optimal walks and update the betweenness estimates. Note that to compute the shortest $\delta$-restless temporal walks we can use Algorithm \ref{alg:truncatedstpaths} provided that we add the condition $t' -t \leq \delta$ in line \ref{alglinetr:forOutNeighTemp}. 

More interestingly, the biggest computational problem arises when updating the temporal betweenness values of the various nodes on the optimal walks. Note that, to do so, we cannot use Algorithm \ref{alg:updatepathcounts} because it does not account for cycles (i.e, when vertices appear multiple times across a walk).
We therefore introduce Algorithm \ref{alg:updatewalkcounts} (pseudocode in Appendix \ref{app:RTPPseudocode}) that works in the presence of cycles. 
The main intuition behind Algorithm \ref{alg:updatewalkcounts} is that we need to recreate backwards all the optimal walks obtained through the RTP version of Algorithm \ref{alg:truncatedstpaths}.
For each walk we will maintain a set that keeps track of the nodes already visited up to the current point of the exploration of the walk, updating a node's estimate if and only if we see such node for the first time. 
This is based on the simple observation that a cycle cannot alter the value of the betweenness centrality of a node on a fixed walk, allowing us to account only once for the node's appearance along the walk. 

We now describe Algorithm \ref{alg:updatewalkcounts} by discussing its differences with Algorithm \ref{alg:updatepathcounts}.
In line \ref{alglineupwa:initdatastruct}, instead of maintaining a matrix keeping track of the presence of a VA in the queue, we now initialize a matrix $\mathsf{u}[\cdot,\cdot]$ that keeps the number of times a VA is in the queue.
The queue, initialized in line \ref{alglineupwa:initqueue}, keeps elements of the form $\langle\cdot,\cdot\rangle$, where the first entry is a VA to be explored and the second entry is the set of nodes already visited backwards along the walk leading to such vertex appearance. 
While visiting backwards each walk, we check if the nodes are visited for the first time on such walk: if so, we update the betweenness values by accounting for the number of times we will visit such VA across other walks (lines \ref{alglineupwa:ifnotvisited}-\ref{alglineupwa:updatebetween}). 
Next, we update the set of nodes visited (line \ref{alglineupwa:updateset}). 
Finally, we update the count $\sigma_{w',t'}^z$ of the walks leading from the predecessor $(w',t')$ of the current VA $(w,t)$ to $z$ (line \ref{alglineupwa:updatepathstow}), the number $\mathsf{u}_{w',t'}$ of times such predecessor will be visited (line \ref{alglineupwa:updatevisits}), and enqueue the predecessor $(w',t')$ to be explored, together with the additional information of the set $\mathsf{S}'$ of nodes explored up to that point. 

To conclude, note that Algorithm \ref{alg:updatewalkcounts} is more expensive than Algorithm \ref{alg:updatepathcounts} since it recreates all the optimal walks, 
while Algorithm \ref{alg:updatepathcounts} avoids such step given the absence of cycles.

\begin{figure}[t]
    \centering
        \begin{tabular}{lr}
            \includegraphics[width=.4\linewidth]{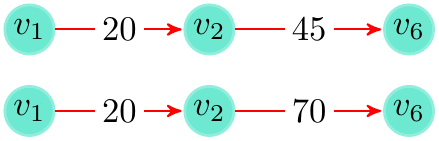} &  \includegraphics[width=.4\linewidth]{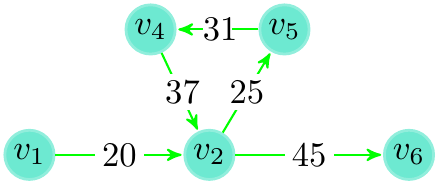} 
        \end{tabular}
        \caption{Considering the temporal network in Figure \ref{fig:basicdef} and $\delta=10$, the paths from node $v_1$ to node $v_6$ on the left are not shortest $\delta$-restless since both violate the timing constraint (i.e., $45-20,70-20 > \delta$). Instead, the walk on the right is shortest and meets the timing constraint with $\delta=10$: so, it is a shortest $\delta$-restless walk.}
    \label{fig:twalks}
\end{figure}

\subsection{\algname\ - Theoretical Guarantees}\label{subsec:theoguar}

In order to address Problem \ref{problem}, \algname\ bounds the deviation between the estimates $\tilde{b}(v)$ and the actual values $b(v)$, for every node $v \in V$. To do so, we leverage on the so called \emph{empirical Bernstein bound}, which we adapted to  \algname.

Given a node $v \in V$, let $\tilde{B}_{v,:} = (\tilde{b}'(v)_1, \tilde{b}'(v)_2, \dots, \tilde{b}'(v)_{\ell})$, where $\tilde{b}'(v)_i$ is the estimate of $b(v)$ by analysing the $i$-th sample, $i \in \{1, \dots, \ell\}$. 
Let $\mathbf{V}(\tilde{B}_{v,:})$ be the \emph{empirical variance} of $\tilde{B}_{v,:}$:
\[
\mathbf{V}(\tilde{B}_{v,:}) = \frac{1}{\ell(\ell-1)} \sum_{1 \leq < i < j \leq \ell} (\tilde{b}'(v)_i-\tilde{b}'(v)_j)^2.
\]

We use the \emph{empirical Bernstein bound} to limit the deviation between $\tilde{b}(v)$'s and $b(v)$'s, which represents Corollary 5 of \cite{maurer2009empirical} adapted to our framework, since Corollary 5 of \cite{maurer2009empirical} is formulated for generic random variables taking values in $[0,1]$ and for an arbitrary set of functions. 

\begin{theorem}[Corollary 5, \cite{maurer2009empirical}]
    \label{th:bound}
    Let $\ell \geq 2$ be the number of samples, and $\eta \in (0,1)$ be the confidence parameter. 
    Let $\tilde{b}'(v)_i$ be the estimate of $b(v)$ by analysing the $i$-th sample, $i \in \{1, \dots, \ell\}$ and $v \in V$. Let $\tilde{B}_{v,:} = (\tilde{b}'(v)_1, \tilde{b}'(v)_2, \dots, \tilde{b}'(v)_{\ell})$, and $\mathbf{V}(\tilde{B}_{v,:})$ be its empirical variance. With probability at least $1-\eta$, and for every node $v \in V$, we have that 
    
    \[
    |\tilde{b}(v) - b(v)| \leq \sqrt{\frac{ 2 \mathbf{V}(\tilde{B}_{v,:}) \ln(4n / \eta)}{\ell}} + \frac{7 \ln(4n / \eta)}{3(\ell - 1)}.
    \]
\end{theorem}

The right hand side of the inequality of the previous theorem differs from Corollary 5 of \cite{maurer2009empirical} by a factor of $2$ in the arguments of the natural logarithms, since in \cite{maurer2009empirical} the bound is not stated in the symmetric form reported in Theorem \ref{th:bound}. Finally, the result about the guarantees on the quality of the estimates provided by \algname\ follows.

\begin{corollary}
    \label{cor:onbra}
    Given a temporal network $T$, the pair $(\varepsilon', \tilde{B}(T))$ in output from \algname\ 
    is such that, with probability $>1-\eta$,  it holds that $\tilde{B}(T)$ is an absolute $(\varepsilon',\eta)$-approximation set of $B(T)$. 
\end{corollary}

Observe that Corollary \ref{cor:onbra} is independent of the structure of the optimal paths considered by \algname, therefore such guarantees hold for both the criteria considered in our work. 

\section{Experimental Evaluation}
\label{sec:exp}
In this section we present our experimental  evaluation that has the following goals: (i)  motivate the study of the temporal betweenness centrality by showing two real world temporal networks on which the temporal betweenness provides novel insights compared to the static betweenness computed on their associated static networks; (ii) assess, considering the STP criterion, the accuracy of the \algname's estimates, and the benefit of using \algname\ instead of the state-of-the-art exact approach~\cite{Buss2020}, both in terms of running time and memory usage; (iii) finally, show how \algname\ can be used on a real world temporal network to analyze the RTP-based betweenness centrality values.

\subsection{Setup}
\begin{table}[t]
	\centering
	\caption{Datasets used and their statistics.}
	\label{tab:datasets}
	\scalebox{0.75}{
		\begin{tabular}{ccccl}
			\toprule
			Name& $n$&$m$& Granularity & Timespan\\
			\midrule
			HighSchool2012 (HS) & 180 & 45K & 20 sec & 7 (days)\\
			CollegeMsg & 1.9K & 59.8K & 1 sec & 193 (days)\\
			EmailEu & 986 & 332K & 1 sec & 803 (days)\\
			FBWall (FB) & 35.9K & 199.8K & 1 sec & 100 (days)\\
			Sms & 44K & 544.8K & 1 sec & 338 (days)\\
			Mathoverflow & 24.8K & 390K & 1 sec & 6.4 (years)\\
			Askubuntu & 157K & 727K & 1 sec & 7.2 (years)\\
			Superuser & 192K & 1.1M & 1 sec & 7.6 (years)\\
			\bottomrule
		\end{tabular}
	}
\end{table}

We implemented \algname\ in C++20 and compiled it 
using \texttt{gcc 9}. The code is freely  available\footnote{\url{https://github.com/iliesarpe/ONBRA}.}. All the experiments were performed sequentially on a 72 core Intel Xeon Gold 5520 @ 2.2GHz machine with 1008GB of RAM available. 
The real world datasets we used are described in Table \ref{tab:datasets}, which are mostly social or message networks from different domains. Such datasets are publicly available online\footnote{\url{http://www.sociopatterns.org/} and \url{https://snap.stanford.edu/temporal-motifs/data.html}.}. For detailed descriptions of such datasets we refer to the links  reported and~\cite{Paranjape2017}. To obtain the FBWall dataset we cut the last 200K edges from the original dataset~\cite{Viswanath2009}, which has more than 800K edges. Such cut is done to allow the exact algorithm to complete its execution without exceeding the available memory.

\subsection{Temporal vs Static Betweenness}
In this section we assess that the temporal betweenness centrality of the nodes of a temporal network provides novel insights compared to its static version. To do so, we computed for two datasets, from different domains, the exact ranking of the various nodes according to their betweenness values. The goal of this experiment is to compare the two rankings (i.e., temporal and static) and understand if the relative orderings are preserved, i.e., verify if the most central nodes in the static network are also the most central nodes in the temporal network. To this end, given a temporal network $T=(V,E)$, let $G_T=(V,\{(u,v) : \exists (u,v,t)\in E \})$ be its associated static network. We used the following two real world networks: (i) HS, that is a temporal network representing a face-to-face interaction network among students; (ii) and FB, that is a Facebook user-activity network \cite{Viswanath2009} (see Table \ref{tab:datasets} for further details). 

We first computed the exact temporal and static betweenness values of the different nodes of the two networks. Then, we ranked the nodes by descending betweenness values. We now discuss how the top-$k$ ranked nodes vary from temporal to static on the two networks. We report in Table \ref{tab:topK} (in Appendix \ref{app:suppldata}) the Jaccard similarity between the sets containing the top-$k$ nodes of the static and temporal networks. On HS, for $k=25$, only 11 nodes are top ranked in both the rankings, which means that less than half of the top-25 nodes are central if only the static information is considered. The value of the intersection increases to $36$ for $k=50$, since the network has only 180 nodes. More interestingly, also on the Facebook network only few temporally central nodes can be detected by considering only static information: only 9 over the top-25 nodes and 15 over the top-50 nodes. 
In order to better visualize the top-$k$ ranked nodes, we show their betweenness values in Figure \ref{subfig:topKvals}: note that there are many top-$k$ temporally ranked nodes having small static betweenness values, and vice versa.

These experiments show the importance of studying the temporal betweenness centrality, which provides novel insights compared to the static version.

\begin{figure}
    \centering
        \begin{tabular}{lr}
        	\subfloat[]{ \includegraphics[width=.46\linewidth]{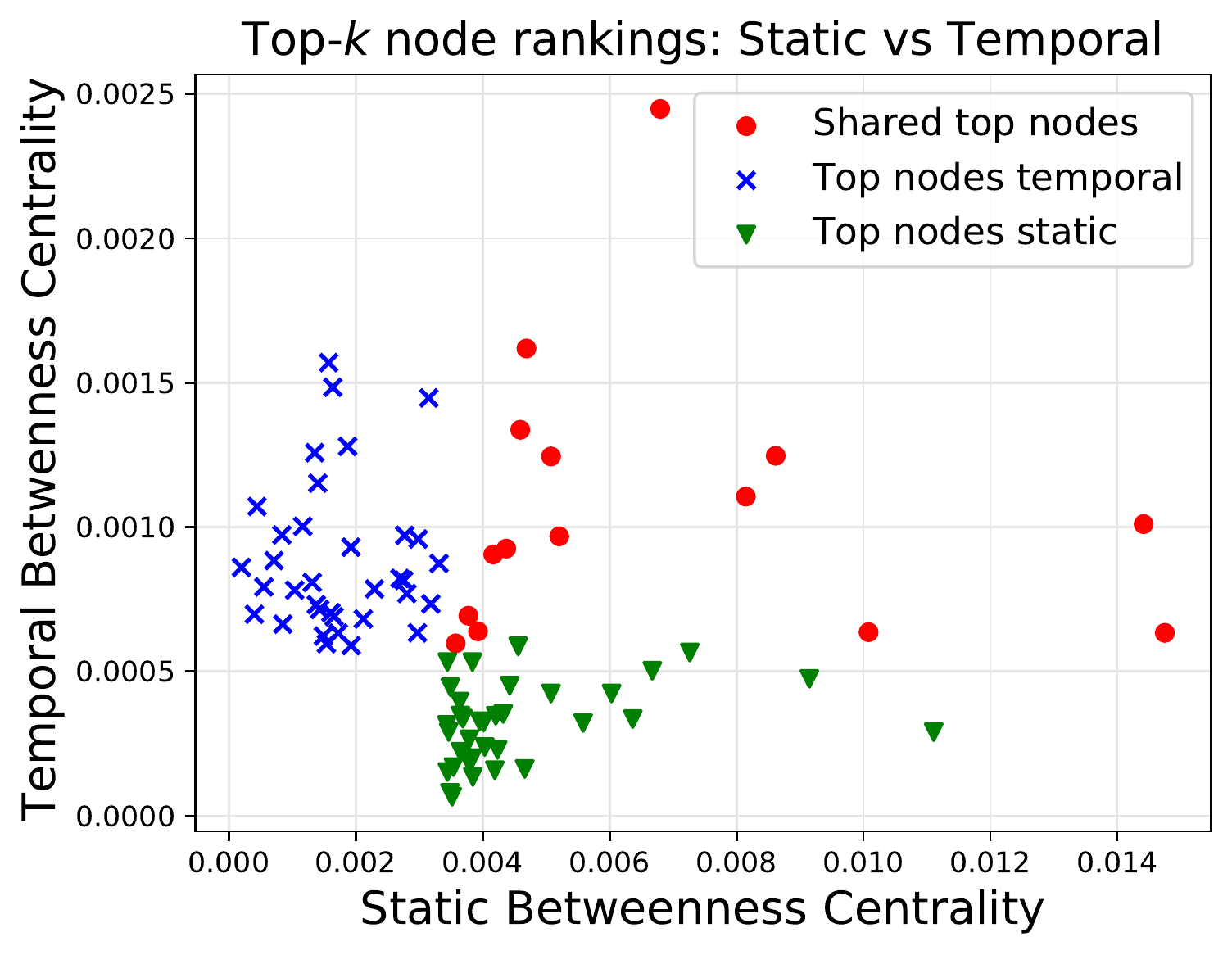} \label{subfig:topKvals}}  
        	\subfloat[]{ \includegraphics[width=.46\linewidth]{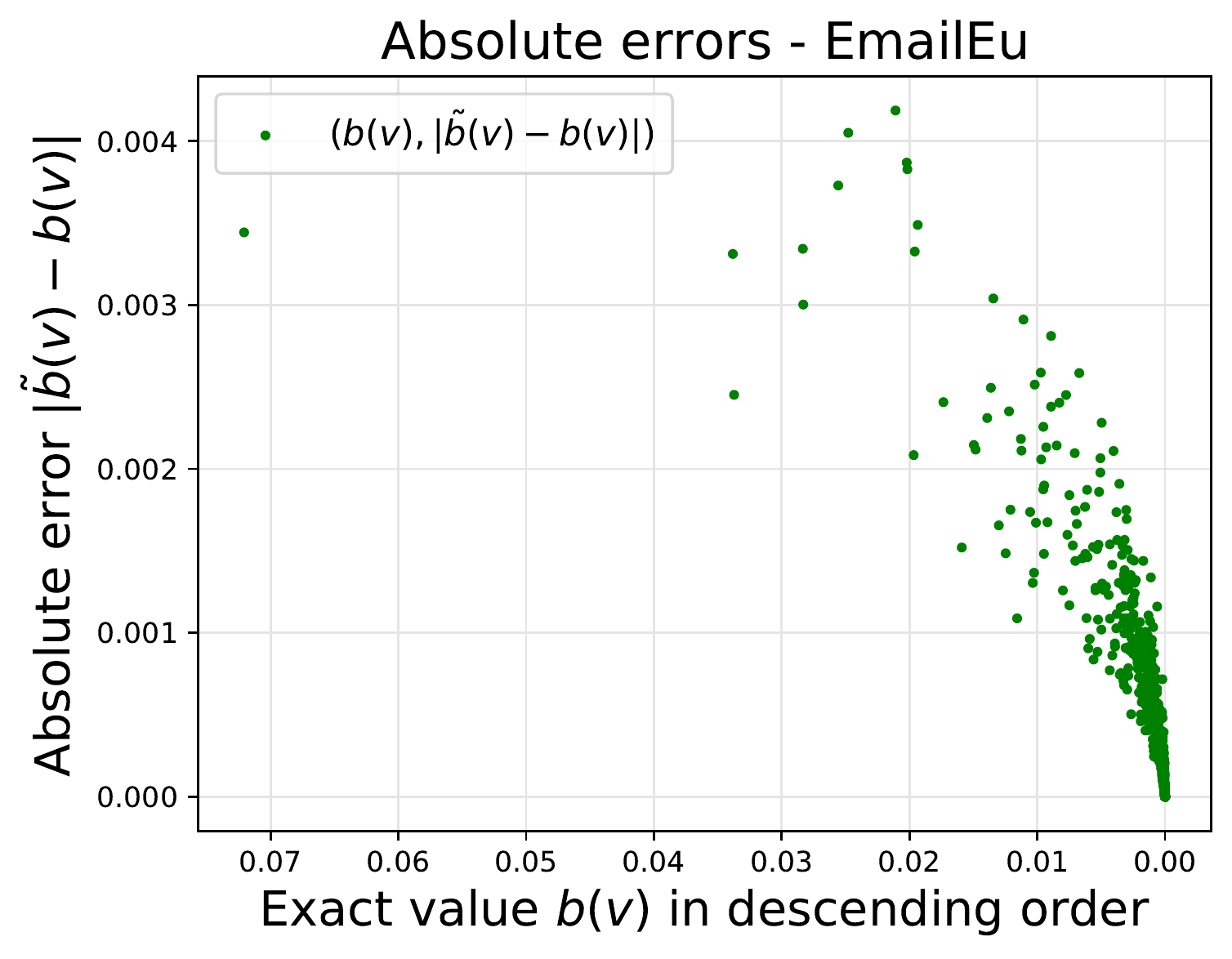} \label{subfig:errors_and_exacts}}\\
        	\subfloat[]{
        	\includegraphics[width=.47\linewidth]{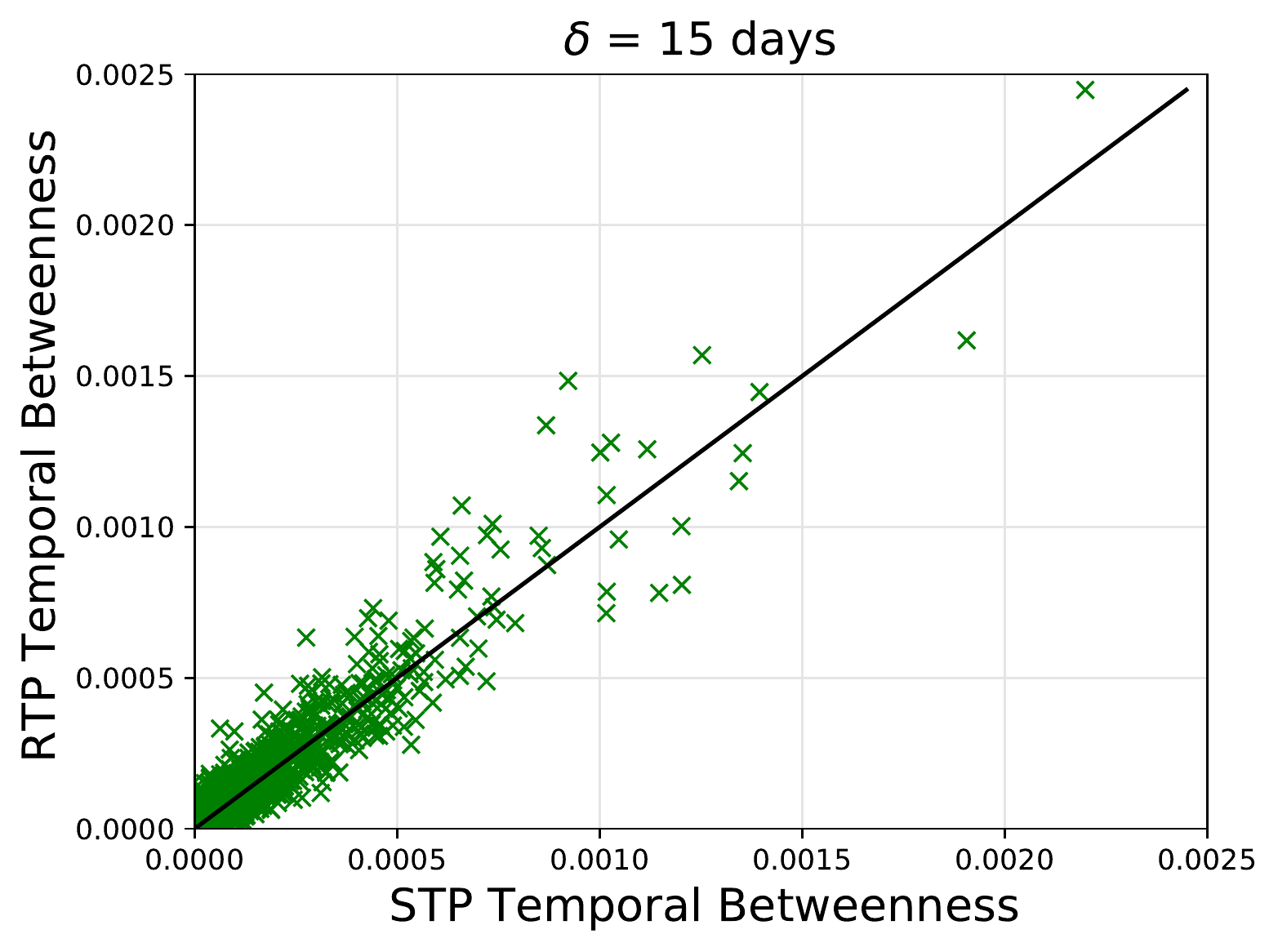}
        	\label{fig:deltavarC}}
        	\subfloat[]{
        	\includegraphics[width=.47\linewidth]{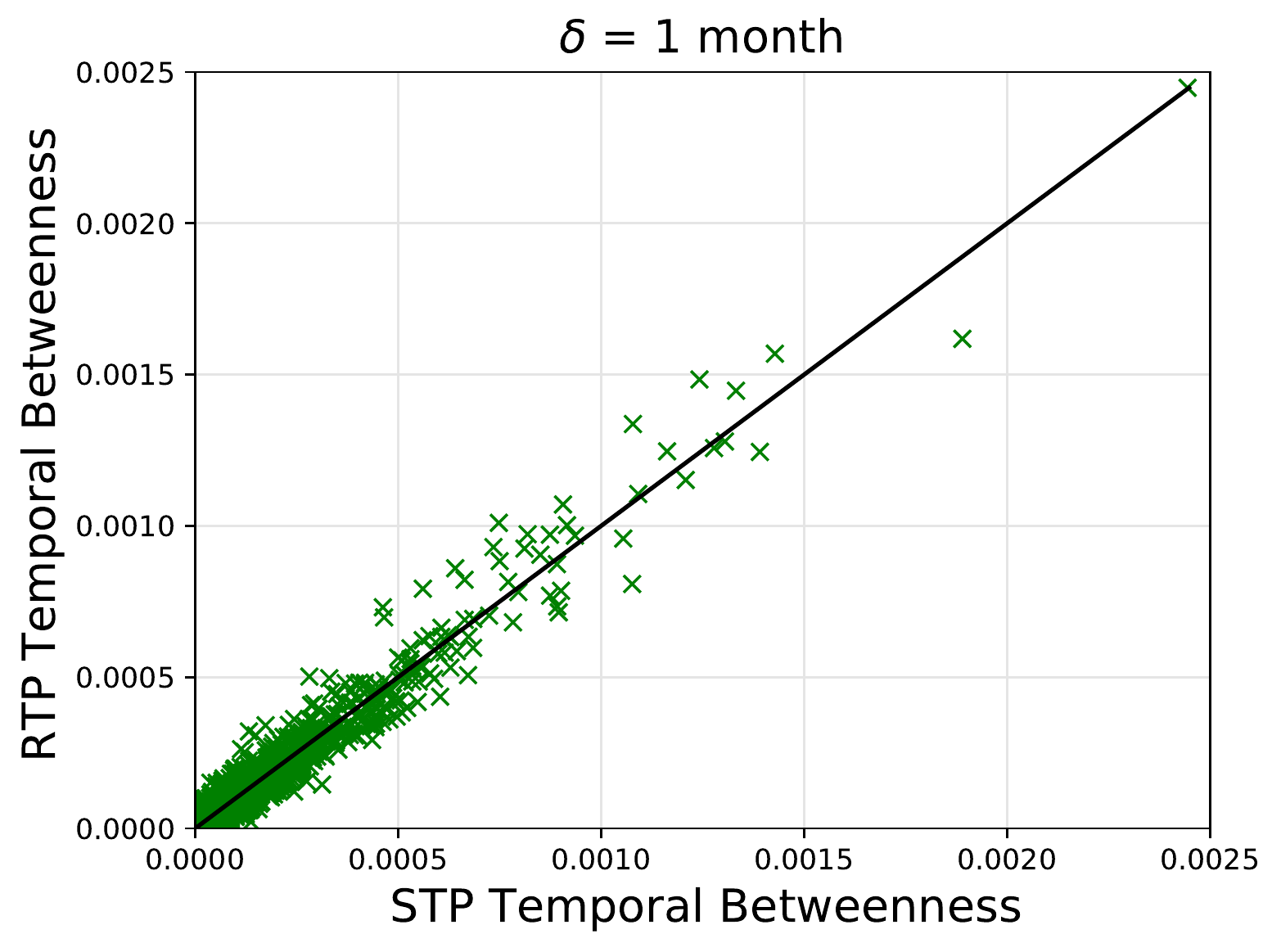}
        	\label{fig:deltavarD}}
        \end{tabular}
    \caption{ (\ref{subfig:topKvals}): static and temporal betweenness values of the top-50 ranked nodes of the dataset FB; (\ref{subfig:errors_and_exacts}): for dataset \texttt{EmailEu}, the deviations (or absolute errors) $|\tilde{b}(v)- b(v)|$ between the estimates $\tilde{b}(v)$ and the actual values $b(v)$ of the temporal betweenness centrality, for decreasing order of $b(v)$; 
     (\ref{fig:deltavarC},\ref{fig:deltavarD}): comparison between the temporal betweenness values based on STP and RTP, for  $\delta$=15 days (left) and $\delta$=1 month (right).}
\end{figure}

\begin{table*}[t]
	\caption{For each dataset, the average and maximum deviation between the estimate $\tilde{b}(v)$ and the actual temporal betweenness value $b(v)$ over all nodes $v$ and $10$ runs, respectively $Avg. \ Error$ and $\sup_{v\in V}|b(v)-\tilde{b}(v)|$, the theoretical upper bound $\varepsilon'$, the $Sample \ rate$ (\%) of pairs of nodes we sampled, the running time $t_{EXC}$ and peak RAM memory $MEM_{EXC}$ required by the exact approach \cite{Buss2020}, the running time $t_{\text{\algname}}$  and peak RAM memory $MEM_{\text{\algname}}$ required by \algname. The symbol \ding{55} denotes that the exact computation of \cite{Buss2020} is not able to conclude on our machine.}
	\label{tab:results}
	\centering
	\scalebox{0.83}{
		\begin{tabular}{ccccccc|cc}
			\toprule
			Dataset & Avg.\ Error & $\sup_{v\in V}|b(v)-\tilde{b}(v)|$ & $\varepsilon'$ & Sample rate (\%) & $t_{\text{EXC}}$ (sec) & $t_{\text{\algname}}$ (sec) & MEM$_{\text{EXC}}$ (GB) & MEM$_{\text{\algname}}$ (GB)\\
			\midrule
			\texttt{CollegeMsg} & $1.74 \cdot 10^{-4}$ & $6.38 \cdot 10^{-3}$ & $2.27 \cdot 10^{-2}$ & 0.083 & 231 & \textbf{148} & 12.0 & \textbf{0.13} \\ 
			\texttt{EmailEu} & $4.69 \cdot 10^{-4}$ & $1.35 \cdot 10^{-2}$ & $6.15 \cdot 10^{-2}$ & 0.093 & 7211 & \textbf{1808} & 23.9 & \textbf{2.1}  \\ 
			\texttt{Mathoverflow} & $6.35 \cdot 10^{-6}$ & $2.1 \cdot 10^{-3}$ & $5.38 \cdot 10^{-3}$ & 0.005 & 79492 & \textbf{36983} & 1004.3 & \textbf{6.8} \\ 
			\texttt{FBWall} & $4.25 \cdot 10^{-6}$ & $5.89 \cdot 10^{-4}$ & $2.13 \cdot 10^{-3}$ & 0.003 & 11489 & \textbf{3145} & 738.0 & \textbf{11.1} \\ 
			\texttt{Askubuntu} & \ding{55} & \ding{55} & $6.92 \cdot 10^{-3}$ & 0.00006 & \ding{55} & \textbf{35585} & $>$1008 & \textbf{20.3}\\ 
			\texttt{Sms} & \ding{55} & \ding{55} & $1.54 \cdot 10^{-3}$ & 0.00231 & \ding{55} & \textbf{13020}&$>$1008 & \textbf{16.2}\\ 
			\texttt{Superuser} & \ding{55} & \ding{55} & $1.02 \cdot 10^{-2}$ & 0.00003 & \ding{55} & \textbf{41856} &$>$1008 & \textbf{16.7}\\  
			\bottomrule
		\end{tabular}
	}
\end{table*}

\subsection{Accuracy and Resources of \algname}

In this section we first assess the accuracy of the estimates $\tilde{B}(T)$ provided by \algname\ considering only the STP criterion, since for the RTP criterion no implemented exact algorithm exists. Then, we show the reduction of computational resources induced by \algname\ compared to the exact algorithm in~\cite{Buss2020}. 

To assess \algname's accuracy and its computational cost, we used four datasets, i.e., \texttt{CollegeMsg},  \texttt{EmailEu}, \texttt{Mathoverflow}, and \texttt{FBWall}. We first executed the exact algorithm, and then we fix $\eta=0.1$ and $\ell$ properly for \algname\ to run within a fraction of the time required by the exact algorithm. The results we now present, which are described in detail in Table \ref{tab:results}, are all averaged over 10 runs (except for the RAM peak, which is measured over one single execution of the algorithms). 

Remarkably, even using less than $1\%$ of the overall pairs of nodes as sample size, \algname\ is able to estimate the temporal betweenness centrality values with very small average deviations between $4 \cdot 10^{-6}$ and $5 \cdot 10^{-4}$, while obtaining a significant running time speed-up between $\approx$$1.5\times$ and $\approx$$4\times$ with respect to the exact algorithm \cite{Buss2020}. 
Additionally, the amount of RAM memory used by \algname\ is significantly smaller than the exact algorithm in~\cite{Buss2020}: e.g., on the \texttt{Mathoverflow} dataset \algname\ requires only $6.8$ GB of RAM peak, which is $147\times$ less than the $1004.3$ GB required by the exact state-of-the-art algorithm~\cite{Buss2020}. 
Furthermore, in all the experiments we found that the maximum deviation is distant at most one order of magnitude from the theoretical upper bound $\varepsilon'$ guaranteed by Corollary \ref{cor:onbra}. Surprisingly, for two datasets (\texttt{EmailEu} and \texttt{Mathoverflow}) the maximum deviation and the upper bound $\varepsilon'$ are even of the same order of magnitude. Therefore we can conclude that the guarantees provided by Corollary \ref{cor:onbra} are often very sharp. In addition, \algname's accuracy is demonstrated by the fact that the deviation between the actual temporal betweenness centrality value of a node and its estimate obtained using \algname\ is about one order of magnitude less than the actual value, as we show in Figure \ref{subfig:errors_and_exacts} and Figure \ref{fig:fig_appendix} (in Appendix \ref{app:suppldata}).

Finally, we show in Table \ref{tab:results} that on the large datasets \texttt{Asku\-buntu}, \texttt{Sms}, and 	\texttt{Superuser} the exact algorithm \cite{Buss2020} is not able to conclude the computation on our machine (denoted with \ding{55}) since it requires more than 1008GB of RAM. Instead, \algname\ provides estimates of the temporal betweenness centrality values in less than $42$K (sec) and $21$ GB of RAM memory. 

To conclude, \algname\ is able to estimate the temporal betweenness centrality with high accuracy providing rigorous and sharp guarantees, while significantly reducing the computational resources required by the exact algorithm in \cite{Buss2020}. 

\subsection{\algname\ on RTP-based Betweenness}
In this section we discuss how \algname\ can be used to analyze real world networks by estimating the centrality values of the nodes for the temporal betweenness under the RTP criterion.

We used the FB network, on which we computed a tight approximation of the temporal betweenness values ($\varepsilon'<10^{-4}$) of the nodes for different values of $\delta$, i.e., $\delta$=1 day, $\delta$=15 days, and $\delta$=1 month. For $\delta$=1 day, we found only 4 nodes with temporal betweenness value different from 0, which is surprising since it highlights that the information spreading across wall posts through RTPs in 2008 on Facebook required more than 1 day of time between consecutive interactions (i.e., slow spreading). We present the results for the other values of $\delta$ in Figures \ref{fig:deltavarC} and \ref{fig:deltavarD}, comparing them to the (exact) STP-based betweenness. Interestingly, 15 days are still not sufficient to capture most of the betweenness values based on STPs of the different nodes, while with $\delta$=1 month the betweenness values are much closer to the STP-based values. While this behaviour is to be expected with increasing $\delta$, finding such values of $\delta$ helps to better characterize the dynamics over the network.

To conclude, \algname\ also enables novel analyses that cannot otherwise be performed with existing tools.
\section{Discussion}
In this work we presented \algname, the first algorithm that provides high-quality approximations of the  temporal betweenness centrality values of the nodes in a temporal network, with rigorous probabilistic guarantees. \algname\ works under two different optimality criteria for the paths on which the temporal betweenness centrality is defined: shortest and restless temporal paths (STP, RTP) criteria. To the best of our knowledge, \algname\ is the first algorithm enabling a practical computation under the RTP criteria.
Our experimental evaluation shows that 
\algname\ provides high-quality estimates with tight guarantees, while  remarkably reducing the computational costs compared to the state-of-the-art in \cite{Buss2020}, enabling analyses that would not otherwise be possible to perform.

Finally, several interesting directions could be explored in the future, such as dealing with different optimality criteria for the paths, and employing sharper concentration inequalities to provide tighter guarantees on the quality of the estimates. 
\begin{acks}
Part of this work was supported by the Italian Ministry of Education, University and Research (MIUR), under PRIN Project n. 20174LF3T8 \enquote{AHeAD} (efficient Algorithms for HArnessing networked Data) and the initiative \enquote{Departments of Excellence} (Law 232/2016), and by University of Padova under project \enquote{SID 2020: RATED-X}.
\end{acks}

\newpage

\bibliographystyle{ACM-Reference-Format}
\bibliography{biblio}


\newpage
\appendix
\section{Missing Proofs}\label{app:proofs}
\begin{lemma} 
	\label{lemma:unbiased}
	Let $v \in V$, then $\tilde{b}'(v)$ is an unbiased estimator of $b(v)$.
\end{lemma}
\begin{proof}
	Let $X_{s z}$ be a Bernoulli random variable that takes value $1$ if the pair of nodes $(s,z)$ is sampled, and $0$ otherwise. Since $\mathbb{E}[X_{sz}]= 1/(n(n-1))$, then by the linearity of expectation,
	\[
	\mathbb{E}[\tilde{b}'(v)] = \frac{1}{n(n-1)} \sum_{s,z\in V\\s\neq z} \frac{\sigma_{s,z}^{sh} (v)}{\sigma_{s,z}^{sh} } \frac{\mathbb{E}[X_{sz}]}{1/n(n-1)}= b(v).
	\]
\end{proof}

\section{Missing Algorithms}\label{app:RTPPseudocode}
In this Section we present Algorithm \ref{alg:updatewalkcounts}, used to compute the temporal betweenness values estimates of the various nodes under the RTP criterion. This is discussed in details in Section \ref{sec:RTP_criteria}.
\begin{algorithm}[h]
    \DontPrintSemicolon
    \SetKwComment{Comment}{$\triangleright$\ }{}
    \LinesNumbered 
    \KwIn{$\tilde{B},i$.}
    \For{$(u,v,t)\in E$}
    {
        $\sigma_{v,t}^z\gets 0$; $\mathsf{u}_{v,t}\gets 0$
        \label{alglineupwa:initdatastruct}\;
    }
    $R \gets$ empty queue; 
    \label{alglineupwa:initqueue}\;
    \ForEach{$t : \sigma_{z,t} > 0$\label{alglineupwa:fortReachingZ}}
    {
        
        \For{$(w,t')\in P_{z,t}$\label{alglineupwa:forPredecessorsZ}}
        {
            $R$.enqueue($\langle(w,t'), \{z\}\rangle$); 
            \label{alglineupwa:initVappeareance}\;           
        }
    }
    \While{$!R.$\emph{empty}$()$\label{alglineupwa:mainwhile}}{
        $\langle(w,t),\mathsf{S}\rangle \gets R.$dequeue$()$\label{alglineupwa:dequeue}\;
        \If{$w\neq s$\label{alglineupwa:nots}}{
            \If{$w\notin \mathsf{S}$\label{alglineupwa:ifnotvisited}}
            {
                $\tilde{B}_{w,i} \gets \tilde{B}_{w,i} + \sigma_{w,t}^z \cdot \sigma_{w,t} / (\sigma_z \cdot \mathsf{u}_{w,t})$\label{alglineupwa:updatebetween}\;
            }
            $\mathsf{S}' \gets \mathsf{S} \cup \{w\}$\label{alglineupwa:updateset}\;
            \For{$(w',t')\in P_{w,t}$\label{alglineupwa:forPred}}
            {
                $\sigma_{w',t'}^z \gets \sigma_{w',t'}^z + \sigma_{w,t}^z/\mathsf{u}_{w,t}$ \label{alglineupwa:updatepathstow}\;
                $\mathsf{u}_{w',t'} \gets \mathsf{u}_{w',t'}+1$\label{alglineupwa:updatevisits}\;
                
                $R$.enqueue($\langle(w',t'), \mathsf{S}'\rangle$); 
                \label{alglineupwa:enqueue}\;           
            }
        }
        
    }

    \caption{Update betweenness estimates - RTP.}
    \label{alg:updatewalkcounts}
\end{algorithm}

\section{Supplementary Data}\label{app:suppldata}
\begin{figure}[b]
	\centering
	\begin{tabular}{lr}
		\subfloat[]{ \includegraphics[width=.49\linewidth]{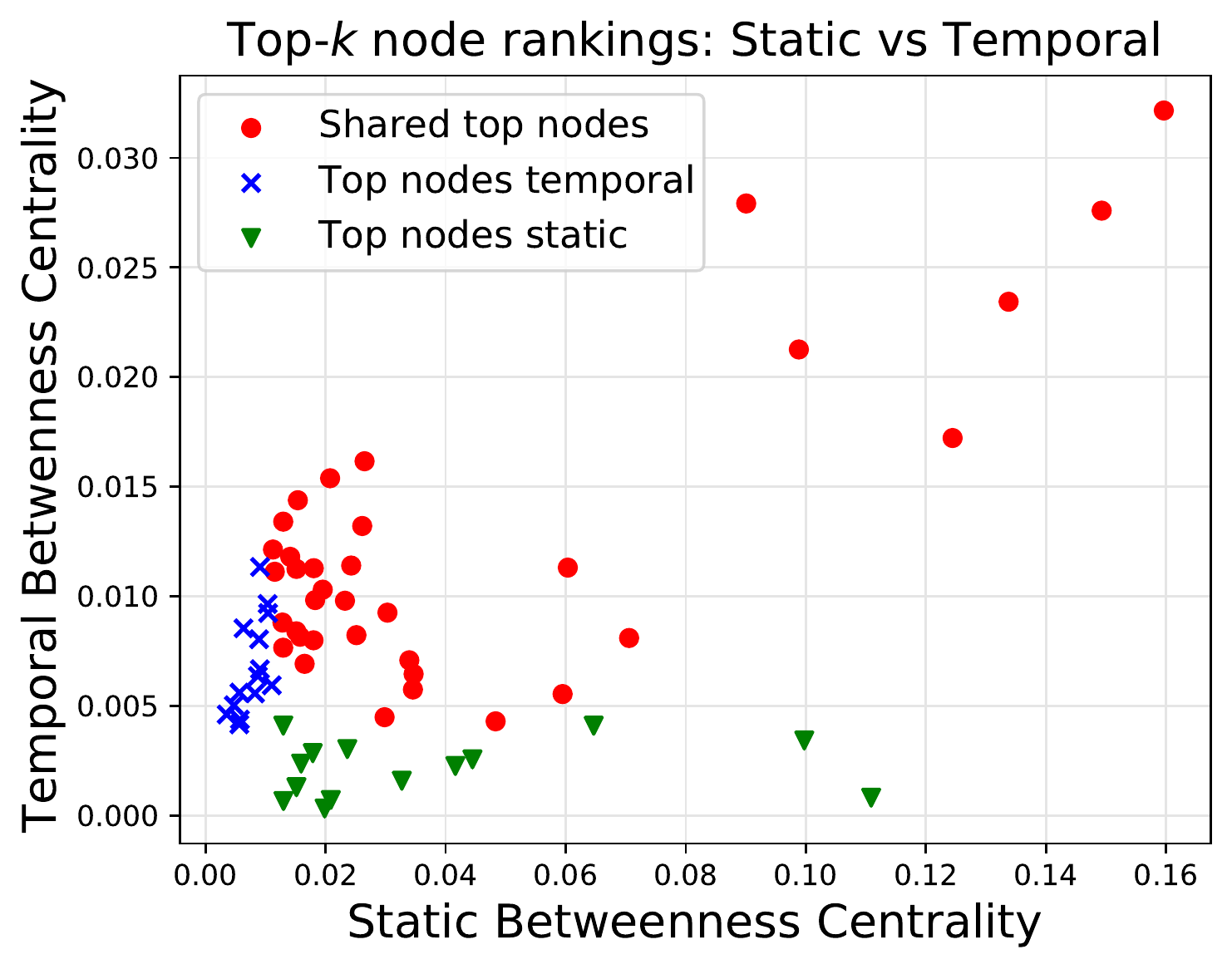} \label{subfig:topKvals_appendix}} \subfloat[]{\includegraphics[width=.49\linewidth]{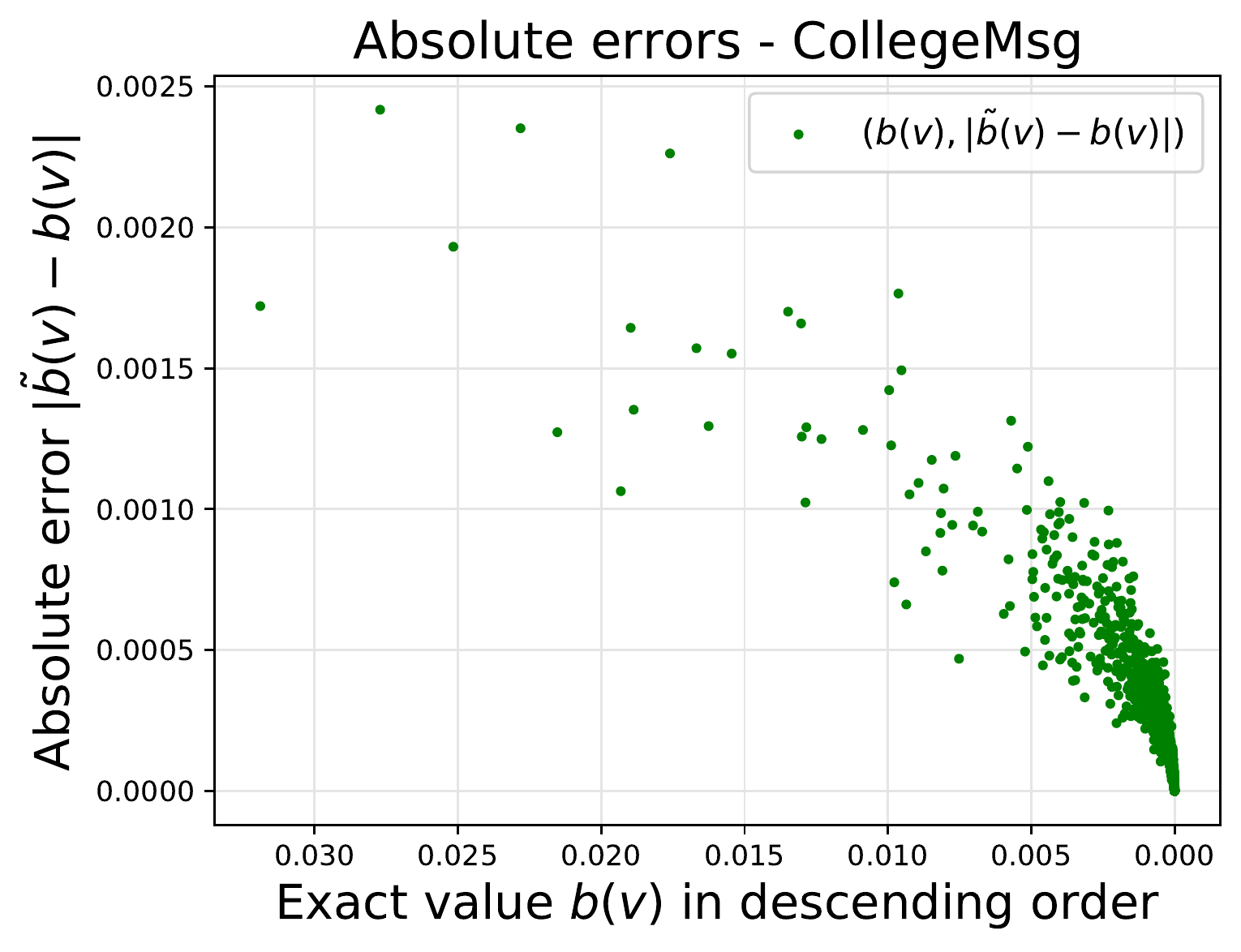} \label{subfig:errors_and_exacts_b_appendix}} \\
		\subfloat[]{\includegraphics[width=.49\linewidth]{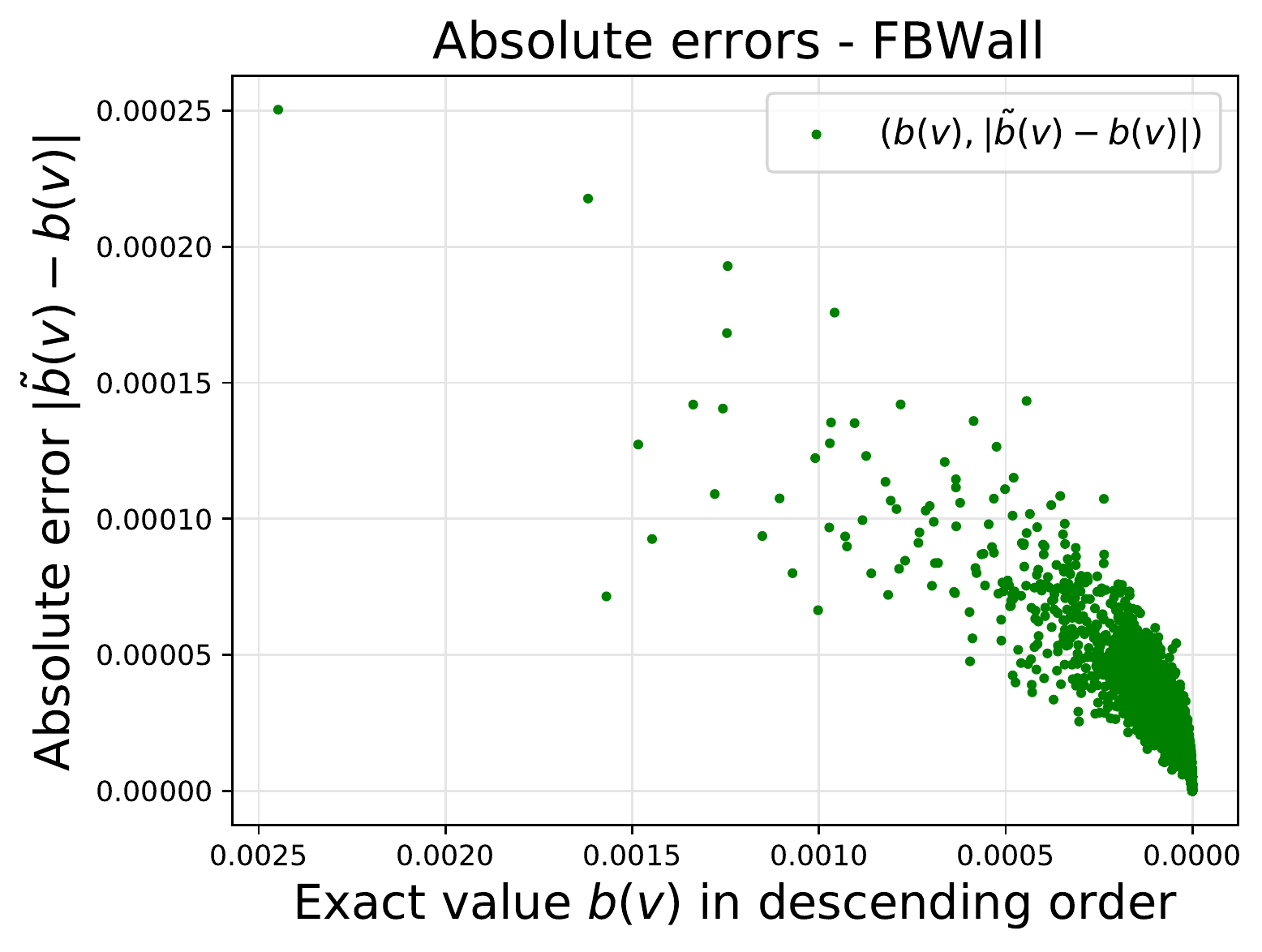} \label{subfig:errors_and_exacts_c_appendix}} \subfloat[]{\includegraphics[width=.49\linewidth]{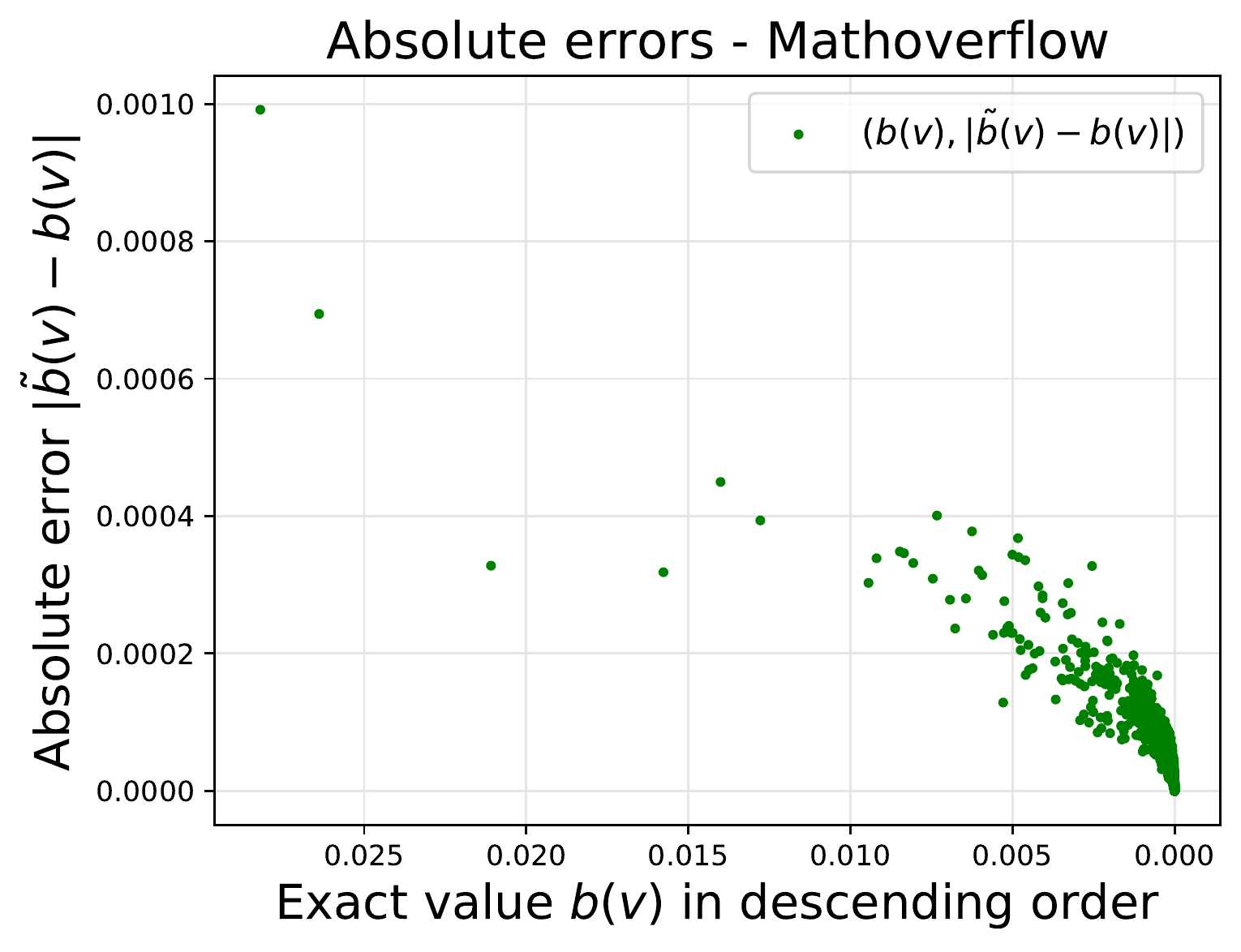} \label{subfig:errors_and_exacts_d_appendix}}
	\end{tabular}
	\caption{(\ref{subfig:topKvals_appendix}): static and temporal betweenness values of the top-50 ranked nodes of the dataset dataset; (\ref{subfig:errors_and_exacts_b_appendix}),(\ref{subfig:errors_and_exacts_c_appendix}), and (\ref{subfig:errors_and_exacts_d_appendix}): respectively for datasets \texttt{CollegeMsg}, \texttt{FBWall}, and  \texttt{Mathoverflow}, the deviations (or absolute errors) $|\tilde{b}(v)- b(v)|$ between the estimates $\tilde{b}(v)$ and the actual values $b(v)$ of the temporal betweenness centrality, for decreasing order of $b(v)$. }
	\label{fig:fig_appendix}
\end{figure}

Given $T$ and $G_T$, let $S_T^k = \{ v_1, \dots, v_k\}$  be the top-$k$ nodes ranked by their temporal betweenness values and let $S_{G_T}^k = \{ v_1', \dots, v_k'\}$  be the top-$k$ nodes ranked by their static betweenness values. We report in Table \ref{tab:topK} the Jaccard similarity $J(k) = |S_T^k \cap S_{G_T}^k|  / |S_T^k \cup S_{G_T}^k|$ for two different values of $k$.
\begin{table}[h]
    \centering
    \caption{Static vs temporal top-$k$ nodes Jaccard similarity $J(k)$. We also report the size of the intersection. }
    \label{tab:topK}
    \scalebox{0.9}{
        \begin{tabular}{ccc}
            \toprule
            Name& $J(25)$ & $J(50)$ \\
            \midrule
            HS & 0.28 (11) & 0.56 (36) \\
            FB & 0.22 (9) & 0.18 (15) \\
            \bottomrule
        \end{tabular}
    }
\end{table}

\end{document}